\documentclass[11pt,a4paper]{article}
\usepackage{indentfirst,mathrsfs}
\usepackage{amsfonts,amsmath,amssymb,amsthm}
\usepackage{latexsym,amscd,multirow}
\usepackage{algorithm}
\usepackage{algpseudocode}
\usepackage{amsbsy}
\usepackage{color}
\usepackage{diagbox}
\usepackage{colortbl}
\newtheorem{thm}{Theorem}
\newtheorem{lem}{Lemma}
\newtheorem{cor}{Corollary}

\newtheorem{prop}{Proposition}

\newtheorem{defn}{Definition}
\newtheorem{example}{Example}
\newtheorem{remark}{Remark}

\newcounter{alg}
\newlength{\lefttab}
\newlength{\numberoffset}
{\trivlist
   \topsep=0pt\parsep=0pt\itemsep=0pt
   \addtolength{\lefttab}{1.25em}
   \addtolength{\numberoffset}{1.25em}
   \leftskip=\lefttab}%
  {\endtrivlist}

\begin{document}
\title{Binary sequences with length $n$ and nonlinear complexity not less than $n/2$}

\author{
Sicheng Liang
\thanks{S. Liang, X. Zeng and Z. Sun are with Faculty of Mathematics and Statistics, Hubei Key Laboratory of
Applied Mathematics, Hubei University, Wuhan 430062, Hubei, China. Email:
sichengliang@aliyun.com, xzeng@hubu.edu.cn, zmsun@hubu.edu.cn},
Xiangyong Zeng
%\thanks{%S. Liang, X. Zeng and Z. Sun are with Faculty of Mathematics and Statistics, Hubei Key Laboratory of Applied Mathematics, Hubei University, Wuhan 430062, Hubei, China. Email: sichengliang@aliyun.com, xzeng@hubu.edu.cn, zmsun@hubu.edu.cn}
,
Zibi Xiao
\thanks{Z. Xiao is with College of Science, Wuhan University of Science and Technology, Wuhan 430081, Hubei, China. Email: xiaozibi@wust.edu.cn.},
Zhimin Sun
%\thanks{}%Z. Sun
%{Yupeng Jiang}
%\thanks{State Key Laboratory of Information Security, Institute of Information Engineering,
%Chinese Academy of Sciences, Beijing 10009, China. Email: jiangyupeng@amss.ac.cn.}
}

\date{}

%\date{}

%\begin{document}
%\title{Binary sequences with length $n$ and nonlinear complexity not less than $n/2$}
%\author
%{Sicheng Liang\thanks{E-mail addresses: XXX},  Xiangyong Zeng\thanks{. E-mail addresses: XXX}, Zibi Xiao\thanks{E-mail addresses: xiaozibi@wust.edu.cn},  Zhimin Sun\thanks{E-mail addresses: XXX}\\
%\\
%\normalsize{Faculty of Mathematics and Statistics, Hubei Key Laboratory of }\\
%\normalsize{Applied Mathematics, Hubei University, Wuhan 430062, China}\\
%\normalsize{College of Science, Wuhan University of Science and Technology,}\\
%\normalsize{Wuhan 430081, Hubei, China}\\
%}

\maketitle
\begin{quote}
{\small {\bf Abstract:}} In this paper, the construction of finite-length binary sequences whose nonlinear
complexity is not less than half of the length is investigated. By characterizing the structure of the sequences, an algorithm is proposed to generate all binary sequences with length $n$ and nonlinear complexity $c_{n}\geq n/2$, where $n$ is an integer larger than $2$. Furthermore, a formula is established to calculate the exact number
of these sequences. The distribution of nonlinear complexity for these sequences is thus completely determined.


%{\small {\bf MSC:}} xxxxx, xxxxx, xxxxx
\end{quote}
\section{Introduction}
Pseudorandom sequences generated by feedback shift registers (FSRs) are widely used in secure
communications \cite{DING1,GOLOMB,MENE}. In cryptographic applications, the sequences are required to be random
or unpredictable. The linear complexity is a classical criterion for
evaluating the unpredictability of a sequence. It measures the length of the shortest
FSRs with linear feedback functions that can generate the sequence. To resist the attack from the
application of the Berlekamp-Massey algorithm \cite{JLM}, the sequences used in cipher systems
should have large linear complexity. A sequence with large linear complexity may
be generated by a much shorter FSR if the shift registers with nonlinear feedback functions are considered. Removal of the restriction on the degree of feedback functions gives the notion of nonlinear complexity, also referred to as the maximum order complexity and the nonlinear span, of a sequence \cite{Jansen1,Jansen2}. Name-wise we use the name nonlinear complexity in this paper since it is more easily seen as the counterpart of linear complexity.

The eSTREAM project (2004-2008) was a research project of European cryptographers to identify a portfolio of new stream ciphers. Some hardware-oriented finalists of the eSTREAM, such as Trivium \cite{TRIVI} and Grain \cite{GRAIN}, were designed based on nonlinear feedback shift registers. This is an important reason that the nonlinear complexity, as a criterion for assessing the randomness of sequences, has attracted more attention in recent years.
%{\color{blue}To make it of more practical cryptographic interest,
%{\color{red}Some attempts have been made to} find the shortest FSRs with (quadratic) nonlinear feedback functions that generate a given sequence \cite{CHAN,LKK2,RK,RPNN,GONG1} {\color{red}and to} establish theoretical bounds of nonlinear complexity \cite{GUTI,KHAC,HX}.
Many authors concentrated on finding the shortest FSRs with (quadratic) nonlinear feedback functions
that generate a given sequence \cite{CHAN,LKK2,RK,RPNN,GONG1}, as well as establishing theoretical bounds of nonlinear
complexity \cite{GUTI,KHAC,HX}. %{\color{red}On the other hand,}
To further measure the randomness of sequences, an approximate probability distribution
for the nonlinear complexity of random binary sequences was derived in \cite{MANN}. The nonlinear complexity
distribution of finite-length binary sequences and periodic binary sequences were investigated in \cite{RK} and \cite{MYKK}, respectively.

Constructing periodic or finite-length sequences with large nonlinear complexity is another important topic in the research field of nonlinear complexity. In reference \cite{PR},
two methods were proposed to construct periodic binary sequences with given linear complexity and maximum
nonlinear complexity. Recently, recursive approaches were proposed to generate all binary
sequences of period $n$ with nonlinear complexity $n-1$ and $n-2$ in \cite{SZLH} and \cite{XZLY}, respectively.
By applying a combinatorial method, all $m$-ary sequences of length $n$ and nonlinear complexity $n-j$,
where $2\leq j\leq4$ and $n\geq 2j$,
%and binary sequences of length $n$ with nonlinear complexity $n-4$
were characterized in \cite{PZS, YILIN}. In addition, several constructions of finite-length sequences with large nonlinear
complexity profile from function fields were studied in \cite{CAST,LXY,HX}.

The purpose of this paper is to investigate finite-length binary sequences of length $n$ and large nonlinear complexity,
where $n$ is a positive integer larger than $2$. In particular, we focus on the sequences with nonlinear complexity $c_n\geq n/2$.
For such a sequence, it is proved in this paper that there is exactly one pair of identical subsequences of length $c_n-1$ with different successors.
This is a refinement of Proposition 3.19 in \cite{Jansen1} (see Lemma \ref{lem1} (iii) in this paper).
With this interesting property, we can divide all these sequences
into $n-c_n$ disjoint classes according to the \textit{distance} $d$ of the unique pair of identical subsequences,
%(referred to as the \textit{distance} of a sequence in Subsection \ref{Sec4.1}),
where $1\leq d\leq n-c_n$. Then we characterize the structure of the sequences in the class with maximum distance $d=n-c_n$.
By studying the nonlinear complexity of left (right) extension
sequences (see Subsection \ref{Sec3.2} for the definition), we find an extension method to
produce all sequences in the classes having distance $d$ with $1\leq d\leq n-c_n-1$.
As a result, we propose an algorithm to generate all sequences of length $n$ and
nonlinear complexity $c_n$ with $c_n\geq n/2$. A formula
to calculate the exact number of these sequences is also established.

There were few constructions of finite-length sequences with large nonlinear complexity in the past.
Moreover, the distribution of the nonlinear complexity of a random sequence is difficult to calculate exactly.
An approximate distribution for nonlinear complexity was given in \cite{MANN}, but the accuracy of the approximation for long sequences is unknown.
In this paper, we not only give a direct construction for all binary sequences of length $n$ with nonlinear complexity $c_n$, where $n/2\leq c_n\leq n-1$, but also present a theoretical result on the number of these sequences.
By applying the established formula, for a very large integer $n$, the exact number of these sequences
 can be calculated, and the exact distribution for nonlinear complexity of binary sequences under the condition $c_n\geq n/2$ is determined. %It is shown that the value

The remainder of this paper is organized as follows. Section \ref{preliminaries} introduces some necessary notations, definitions and properties of the nonlinear complexity. In Section \ref{Sec3}, some useful properties of the nonlinear
complexity of finite-length binary sequences with length $n$ and nonlinear
complexity $c_n\geq n/2$ are presented. In Section \ref{Sec4}, the construction of these sequences is investigated and an algorithm is then proposed to generate them.
In Section \ref{Sec5}, a formula for the exact number of all these sequences is established and
a distribution property of the nonlinear complexity is obtained. Section \ref{Sec6} concludes the study.

\section{Preliminaries}\label{preliminaries}
Throughout this paper, let $q$ be an integer with $q\geq2$ and $\mathbb{Z}_{q}=\{0,1,\cdots,q-1\}$ be the residue ring modulo $q$. Let $\textbf{s}=(s_{0},s_{1},s_{2},\cdots)$ be a sequence over $\mathbb{Z}_q$.
The definition and some useful properties of nonlinear complexity of a sequence $\textbf{s}$ are recalled below.
\begin{defn}(\cite{Jansen1})\label{nlcc}  The nonlinear complexity of a sequence $\textbf{s}$, denoted by $nlc(\textbf{s})$, is the length of the shortest feedback shift registers that can generate the sequence $\textbf{s}$.
\end{defn}
The following basic notations will be used throughout this paper.
\begin{itemize}
  \item $\textbf{s}_{n}=(s_{0},s_{1},\cdots,s_{n-1})$ denotes a finite-length sequence of length $n$ over $\mathbb{Z}_q$;

  \item $S_{i}^{i+m-1}=(s_{i},\cdots,s_{i+m-1})$ denotes the $i$th subsequence of length $m$ in the sequence $\textbf{s}_{n}$;

  \item $(\textbf{s}_{n})^{t}=(s_{0},s_{1},\cdots,s_{n-1})^{t}$ denotes the concatenation of $t$ copies of the sequence $\textbf{s}_{n}$,
  and $\textbf{s}_{n}\textbf{e}_{m}$ denotes the concatenation of two sequences $\textbf{s}_{n}$ and $\textbf{e}_{m}$;

  \item $Z_{q}(n,c_{n})$ denotes the set of sequences of length $n$ and nonlinear complexity $c_{n}$ over $\mathbb{Z}_{q}$;

  \item $|D|$ denotes the cardinality of a set $D$.

\end{itemize}

\begin{lem}(\cite{Jansen1})\label{lem1}
Let $\textbf{s}_{n}=(s_{0},s_{1},\cdots, s_{n-1})$ be a sequence of length $n$ over $\mathbb{Z}_{q}$.

{\rm (i) }
   If $l$ is the length of the longest subsequences of $\textbf{s}_{n}$ that occur at least twice with different successors, then $nlc(\textbf{s}_{n})=l+1$.

{\rm (ii) }
For any sequence $\textbf{s}_{n}$, we have $0\leq nlc(\textbf{s}_{n})\leq n-1$. Moreover, $nlc(\textbf{s}_{n})=0$ iff $\textbf{s}_{n}$ has the form $(\alpha, \alpha, \cdots, \alpha)$,
where $\alpha \in \mathbb{Z}_{q}$; $nlc(\textbf{s}_{n})=n-1$ iff $\textbf{s}_{n}$ has the form $(\alpha, \alpha, \cdots, \alpha,\beta)$, where $\alpha, \beta \in\mathbb{Z}_{q}$ and $\alpha\neq \beta$.

{\rm (iii) }
If $nlc(\textbf{s}_{n})\geq \frac{n}{2}$, then its subsequences of length $nlc(\textbf{s}_{n})$
    are all distinct.
\end{lem}

Let $\textbf{s}_{n}=(s_{0},s_{1},\cdots,s_{n-1})$ and $nlc(\textbf{s}_{n})=c_{n}$. Then the $i$th state vector of the FSR which generates the sequence $\textbf{s}_{n}$, denoted by $\underline{S}^{i}$, is exactly the $i$th subsequence of length $c_n$ in $\textbf{s}_{n}$, i.e., $\underline{S}^{i}=S_{i}^{i+c_{n}-1}=(s_i, s_{i+1},\cdots,s_{i+c_n-1})$. The so-called \textit{state sequence} of $\textbf{s}_{n}$ introduced by Jansen is defined as follows.
\begin{defn}(\cite{Jansen1})\label{defn1}
The state sequence $\underline{S}$ of $\textbf{s}_{n}$ is the sequence of states through which the FSR cycles when generating the sequence $\textbf{s}_{n}$, i.e., $\underline{S}=(\underline{S}^{0},\underline{S}^{1},\cdots,\underline{S}^{n-c_{n}})$.
\end{defn}

The most general form of the state sequence of $\textbf{s}_n$ is
\begin{eqnarray}\label{state}
\underline{S}=(\underbrace{\underline{S}^{0},\underline{S}^{1},\cdots,\underline{S}^{h-1}}_{\text{prefix}} ,\overbrace{\underbrace{\underline{S}^{h},\underline{S}^{h+1},\cdots,\underline{S}^{h+p-1}}_{\text{period}},\cdots,\underline{S}^{h+p-1} }^{\text{$mp$}} ,\underbrace{\underline{S}^{h},\underline{S}^{h+1},\cdots,\underline{S}^{h+k-1}}_{\text{suffix}}),
\end{eqnarray}
where the integers $h$, $p$, $m$ and $k$ satisfy $h\geq0$, $p\geq 1$, $m\geq0$, and $0\leq k\leq p-1$.
In particular, $m=0$ means that all the $n-c_n+1$ states of $\textbf{s}_n$ are distinct. If $m\geq1$,
then all the states appear periodically except for the first $h$ states, which implies that the first $h+p$ states in the state sequence $\underline{S}$ are all distinct. The parameters clearly satisfy the following equation \cite{Jansen1}
\begin{eqnarray}\label{Len1}
n=c_{n}+h+mp+k-1.
\end{eqnarray}

\begin{lem}(\cite{Jansen1})\label{lem0}
Let $\textbf{s}_{n}=(s_{0},s_{1},\cdots, s_{n-1})$ be a
sequence of length $n$ over $\mathbb{Z}_{q}$, and $\textbf{s}_{n+1}$ be the sequence of length $n+1$
obtained by extending $\textbf{s}_{n}$ with a term $s_n\in\mathbb{Z}_{q}$.
Let $nlc(\textbf{s}_{n})=c_n$ and $nlc(\textbf{s}_{n+1})=c_{n+1}$.

{\rm (i) }
    If $c_{n}\geq \frac{n}{2}$, then $c_{n+1}=c_{n}$, that is, the nonlinear complexity remains
    the same as $\textbf{s}_{n}$ is extended with a term $s_n$, regardless of which element of $\mathbb{Z}_{q}$ it is.

{\rm (ii) }
    If $c_{n}< \frac{n}{2}$, then extending the sequence with a term $s_n$ may increase
    the value of the nonlinear complexity to a maximum of $n-c_{n}$,
    i.e., $c_{n}\leq c_{n+1}\leq n-c_{n}$.

%{\rm (iii) }
 %   {\color{red}$c_{n+1}> c_{n}$ if and only if the state sequence $S$ of $\textbf{s}_n$ has the form in (\ref{state}) with $m\geq 1$, and %there exists an integer $j$ with $0\leq j\leq n-c_{n}-1$ such that %$\underline{S}^{n-c_{n}}=\underline{S}^{j}$
 %   $S_{n-c_{n}}^{n-1}=S_{n-c_n-p}^{n-p-1}$ and $s_{n}\neq s_{n-p}$, where $m\geq 1$, $p$ and $k$ are the parameters in the state sequence (\ref{state}) of $\textbf{s}_n$}.
    %$S_{n-c_{n}}^{n-1}=S_{j}^{j+c_{n}-1}$ and $s_{n}\neq s_{j+c_{n}}$.
 %   In addition, $c_{n+1}-c_n=(m-1)p+k$.

%{\rm (iii) }
%    {\color{red}$c_{n+1}> c_{n}$ if and only if the $m\geq 1$, %there exists an integer $j$ with $0\leq j\leq n-c_{n}-1$ such that %$\underline{S}^{n-c_{n}}=\underline{S}^{j}$
 %   $S_{n-c_{n}}^{n-1}=S_{n-c_n-p}^{n-p-1}$ and $s_{n}\neq s_{n-p}$, where $m\geq 1$, $p$ and $k$ are the parameters in the state sequence (\ref{state}) of $\textbf{s}_n$}.
    %$S_{n-c_{n}}^{n-1}=S_{j}^{j+c_{n}-1}$ and $s_{n}\neq s_{j+c_{n}}$.
 %   In addition, $c_{n+1}-c_n=(m-1)p+k$.

{\rm (iii) }
$c_{n+1}> c_{n}$ if and only if %there are two identical state vectors in the state sequence $\underline{S}$ of  $\textbf{s}_{n}$ and
$\textbf{s}_{n}$ is extended such that the uniqueness of successor states is violated, that is, $S_{n-c_{n}}^{n-1}=S_{n-c_n-p}^{n-p-1}$ and $s_{n}\neq s_{n-p}$. In addition, $c_{n+1}-c_n=(m-1)p+k$, where $m$, $p$ and $k$ are the parameters in (\ref{state}) and $m\geq 1$.

\end{lem}

\section{Some useful properties}\label{Sec3}

In this section, we establish a necessary condition for a finite-length
sequence $\textbf{s}_{n}$ with $n\geq 3$ and $nlc(\textbf{s}_{n})\geq\frac{n}{2}$
and obtain some interesting
properties of the nonlinear complexity of left extension sequences. Note that all the sequences
considered in the rest of this paper are over $\mathbb{Z}_{2}$.

\subsection{Characterizations on subsequences of length $nlc(\textbf{s}_{n})-1$}

According to Lemma \ref{lem1} (i), the nonlinear complexity of $\textbf{s}_{n}$ is determined by
its subsequences of length $nlc(\textbf{s}_{n})-1$ with different successors. We first present some properties
of these significant subsequences in a sequence $\textbf{s}_{n}$ with nonlinear
complexity $nlc(\textbf{s}_{n})\geq\frac{n}{2}$.

\begin{prop}\label{lem2} Let $\textbf{s}_{n}=(s_{0},s_{1},\cdots, s_{n-1})$ be a sequence over $\mathbb{Z}_2$ and $nlc(\textbf{s}_{n})=c_{n}\geq \frac{n}{2}$. Then any subsequence of length $c_{n}-1$ with a successor occurs at most twice in $\textbf{s}_{n}$.
\end{prop}

\begin{proof}
The case $n-c_n=1$ is trivial since there are only two subsequences of length $c_n-1$ with a successor in $\textbf{s}_n$. For the case $n-c_n>1$, suppose that there exists a subsequence of length $c_{n}-1$ with a successor that occurs more than twice in $\textbf{s}_n$, without loss of generality, this subsequence can be assumed as
\begin{center}
$S_{i}^{i+c_{n}-2} = S_{j}^{j+c_{n}-2}=S_{t}^{t+c_{n}-2}$,
\end{center}
where the integers $i$, $j$ and $t$ satisfy $0\leq i<j<t\leq n-c_{n}$.

Since $\textbf{s}_{n}$ is a sequence over $\mathbb{Z}_2$, among these three identical subsequences of length $c_{n}-1$, two of them must be followed by an identical term. It means that there exist two identical subsequences of length $c_{n}$ in $\textbf{s}_{n}$, a contradiction to Lemma \ref{lem1} (iii). The claimed statement thus follows.
\end{proof}

%Proposition 1 shows that any subsequence of length $c_n-1$ may appear twice in a sequence of length $n$ and nonlinear complexity $c_n$.
From Lemma \ref{lem1} (i), there is one pair of identical subsequences of length $nlc(\textbf{s}_{n})-1$ with different successors in $\textbf{s}_{n}$
if $nlc(\textbf{s}_{n})>1$. In fact, there is exactly one pair of such subsequences in $\textbf{s}_{n}$ when $nlc(\textbf{s}_{n})\geq\frac{n}{2}$.
To prove this, we need the following lemma.

\begin{lem}\label{lem3}
Let $\textbf{s}_{n}=(s_{0},s_{1},\cdots, s_{n-1})$ be a sequence over $\mathbb{Z}_2$ and $nlc(\textbf{s}_{n})=c_{n}\geq 1$.
Let $\textbf{s}_{n+1}$ be a sequence obtained by extending $\textbf{s}_{n}$ with a term $s_n$ and $nlc(\textbf{s}_{n+1})=c_{n+1}$.
If $c_{n+1}>c_{n}$, then there is only one pair of identical subsequences of length $c_{n+1}-1$ with different successors in $\textbf{s}_{n+1}$.
\end{lem}

\begin{proof}
The existence is assured by Lemma \ref{lem1} (i), and it suffices to prove the uniqueness.
Since $c_{n+1}>c_{n}$, the general form of the state sequence of $\textbf{s}_{n}$ is
\begin{eqnarray*}\label{eq2}
\underline{S}=(\underbrace{\underline{S}^{0},\underline{S}^{1},\cdots,\underline{S}^{h-1}}_{\text{prefix}} ,\overbrace{\underbrace{\underline{S}^{h},\underline{S}^{h+1},\cdots,\underline{S}^{h+p-1}}_{\text{period}},\cdots,\underline{S}^{h+p-1} }^{\text{$mp$}} ,\underbrace{\underline{S}^{h},\underline{S}^{h+1},\cdots,\underline{S}^{h+k-1}}_{\text{suffix}}),
\end{eqnarray*}
where $m\geq 1$, and $\textbf{s}_{n}$ is extended such that the uniqueness of successor states is violated by Lemma \ref{lem0} (iii).
Then, one can verify that for the sequence $s_{n+1}$,
\begin{center}
  $S_{h}^{n-p-1}=S_{h+p}^{n-1}$ \quad \text{and} \quad  $s_{n}\neq s_{n-p}$.
\end{center}
In addition, the length of this pair of subsequences is equal to $n-p-h=c_n+(m-1)p+k-1=c_{n+1}-1$ by equation (\ref{Len1}) and Lemma \ref{lem0} (iii).

Since $\underline{S}^{0}, \underline{S}^{1},\cdots,\underline{S}^{h+p-1}$ are distinct,
it follows that  the first $h+p$ subsequences of length $c_{n+1}-1$ in $\textbf{s}_{n+1}$ are distinct.
It is obvious that there are altogether $n-c_{n+1}+2=h+p+1$ subsequences of length $c_{n+1}-1$ with a successor in $\textbf{s}_{n+1}$,
so that $S_{h}^{n-p-1}$ and $S_{h+p}^{n-1}$ are the only pair of identical subsequences of length $c_{n+1}-1$ with different successors in $\textbf{s}_{n+1}$. The proof is complete.
\end{proof}

By applying Lemma \ref{lem3}, we obtain the following theorem, which plays a significant role in the sequel.

\begin{thm}\label{thm1} Let $\textbf{s}_{n}=(s_{0},s_{1},\cdots, s_{n-1})$ be a sequence over $\mathbb{Z}_2$
and $nlc(\textbf{s}_{n})=c_{n}\geq\frac{n}{2}$. Then there exists exactly one pair of identical subsequences of length $c_{n}-1$ with different successors in $\textbf{s}_{n}$.
\end{thm}
\begin{proof}
Let $\textbf{s}_{j}=(s_{0},s_{1},\cdots, s_{j-1})$ be the subsequence of length $j$ that consists of the first $j$ terms of $\textbf{s}_{n}$, where the integer $j$ satisfies $1\leq j\leq n$. Denote the nonlinear complexity of $\textbf{s}_{j}$ by $c_{j}$. It is obvious that $c_{1}\leq c_{2}\leq\cdots\leq c_{n}$. Let $l$ with $2\leq l\leq n$ be an integer such that $c_{l}=c_{n}$ and $c_{l-1}<c_{l}$. Then it follows from Lemma~\ref{lem0} (i) that $c_{l-1}<\frac{l-1}{2}$.

If $l=n$, we have $c_{n}>c_{n-1}$, then the assertion follows directly from Lemma~\ref{lem3}. It suffices to restrict our attention to the case $l\leq n-1$.

Since $c_{n}=c_{l}>c_{l-1}$, it follows from Lemma~\ref{lem3} that among
$S_{0}^{c_n-2}$, $S_{1}^{c_n-1}$,$\cdots$, $S_{l-c_n}^{l-2}$, the $l-c_{n}+1$ subsequences of length $c_{n}-1$ in $\textbf{s}_{l}$,
there is exactly a pair of subsequences that are identical and have different successors.
For the sequence $\textbf{s}_{n}$, we need to consider the $n-l$ additional subsequences of length $c_{n}-1$ with successors, that is,
\begin{eqnarray}\label{pr}
  S_{l-c_{n}+1}^{l-1},\quad S_{l-c_{n}+2}^{l},\quad\cdots,\quad S_{n-c_{n}}^{n-2},
\end{eqnarray}
and show that there is no other identical subsequence of length $c_{n}-1$ in $\textbf{s}_{n}$.

Suppose that there are other two identical subsequences of length $c_{n}-1$ in $\textbf{s}_{n}$, say
$$S_{a}^{a+c_{n}-2} = S_{b}^{b+c_{n}-2},$$
where the integers $a$ and $b$ satisfy $l-c_{n}+1\leq b\leq n-c_{n}$ and $0\leq a<b$.
Also since $c_{l}>c_{l-1}$, the subsequence $S_{l-c_{l-1}-1}^{l-1}$ appears only once in $\textbf{s}_{l}$.
For otherwise there would exist a pair of identical subsequences of length $c_{l-1}$ with different successors in $\textbf{s}_{l-1}$
 by Lemma~\ref{lem0} (iii), which would contradict $nlc(\textbf{s}_{l-1})=c_{l-1}$.
On the other hand, since $c_{l-1}<\frac{l-1}{2}$, we have $c_{n}=c_{l}\leq l-1-c_{l-1}$ by Lemma \ref{lem0} (ii).
Together with $c_{n}\geq \frac{n}{2}$ and $l\leq n-1$, we get
$$l-c_{l-1}-1\geq c_n\geq n-c_n\geq l-c_n+1.$$
Therefore, we conclude that all the $n-l$ additional sequences in (\ref{pr}) contain $S_{l-c_{l-1}-1}^{l-1}$ as their subsequence.  Then we have
$$S_{e}^{e+c_{l-1}}=S_{l-c_{l-1}-1}^{l-1},$$
where the integer $e=l-c_{l-1}-1-(b-a)<l-c_{l-1}-1$, which contradicts the fact that $S_{l-c_{l-1}-1}^{l-1}$ appears only once in $\textbf{s}_{l}$. The proof is complete.
\end{proof}
\begin{remark}\label{mark1}
This result provides a refinement of Lemma \ref{lem1} (iii).
For an arbitrary sequence $\textbf{s}_{n}$ %over $\mathbb{Z}_2$
with $nlc(\textbf{s}_{n})=c_{n}\geq\frac{n}{2}$, it follows from Lemma \ref{lem1} (iii) that all the subsequences of length $c_n$ are distinct. And if the sequence is over $\mathbb{Z}_2$, %Lemma \ref{lem1} (iii) states that for an arbitrary binary sequence $\textbf{s}_{n}$ with $nlc(\textbf{s}_{n})=c_{n}\geq\frac{n}{2}$, all its subsequences of length $c_n$ are different. And
Theorem~\ref{thm1} further shows that there is only one pair of subsequences of length $c_n$ satisfying that their first $c_n-1$ corresponding terms are identical but their last terms are not equal,
%for an arbitrary  sequence $\textbf{s}_{n}$ over $\mathbb{Z}_2$ with $nlc(\textbf{s}_{n})=c_{n}\geq\frac{n}{2}$,
that is, there exists a unique pair of integers $p_{1}$ and $p_{2}$ with $0\leq p_{1}<p_{2}\leq n-c_{n}$ such that
\begin{eqnarray}\label{eq4}
S_{p_{1}}^{p_{1}+c_{n}-2} = S_{p_{2}}^{p_{2}+c_{n}-2}\quad \text{and}\quad s_{p_{1}+c_{n}-1}\neq s_{p_{2}+c_{n}-1}.
\end{eqnarray}
The two integers $p_{1},p_{2}$ are important parameters with respect to $\textbf{s}_{n}$, and will be used frequently when characterizing the sequence $\textbf{s}_{n}$ in the sequel.

It should be noted that there may exist more than one pair of identical subsequences of length $c_{n}-1$ with different successors in $\textbf{s}_{n}$ if $c_{n}<\frac{n}{2}$. For example, the sequence $\textbf{s}_{8}=(0,0,1,0,1,1,0,0)$ has nonlinear complexity $c_{8}=3<\frac{8}{2}$. However, there exist two pair of identical subsequences of length $2$ with different successors in $\textbf{s}_{8}$, that is,
$$S_{1}^{2} = S_{3}^{4}\,,s_{3}\neq s_{5}\quad \text{and}\quad S_{2}^{3} = S_{5}^{6}\,,s_{4}\neq s_{7}.$$
\end{remark}
\subsection{Properties of left extension sequences}\label{Sec3.2}

The behavior of nonlinear complexity of the sequences obtained by adding terms at the end of $\textbf{s}_{n}$
has been studied in \cite{Jansen1,Jansen2,RK}. In this subsection, we will investigate the nonlinear complexity of the sequences
 obtained by adding terms at the beginning of $\textbf{s}_{n}$. In what follows, we will define the
 {\it $t$-term left (right) extension operation} and {\it $t$-term left (right) extension sequence}.
\begin{defn}\label{defn2}
For a positive integer $n$, let $\textbf{s}_{n}=(s_{0},s_{1},\cdots, s_{n-1})$ be a finite-length sequence over $\mathbb{Z}_2$.
For a positive integer $t$,
the operation of adding $t$ terms at the beginning of the sequence $\textbf{s}_{n}$ is called
$t$-term left extension operation, denoted by $\mathfrak{L}_{t}(\textbf{s}_n)$. If we set the $t$ terms
successively added at the beginning of $\textbf{s}_{n}$ as $\alpha_{0},\alpha_{1},\cdots,\alpha_{t-1}$, then
for a fixed $(\alpha_{0},\alpha_{1},\cdots,\alpha_{t-1})\in\mathbb{Z}_2^{t}$, the sequence
$$L_{t}(\textbf{s}_{n})=(\alpha_{t-1},\cdots,\alpha_{1},\alpha_{0},s_{0},s_{1},\cdots, s_{n-1})$$
is called a $t$-term left extension sequence of $\textbf{s}_{n}$. Similarly, the operation of
adding $t$ terms at the end of $\textbf{s}_{n}$ is called $t$-term right extension operation,
denoted by $\mathfrak{R}_{t}(\textbf{s}_n)$. A $t$-term right extension sequence of $\textbf{s}_{n}$
is accordingly defined as
$$R_{t}(\textbf{s}_{n})=(s_{0},s_{1},\cdots, s_{n-1},\beta_{0},\beta_{1},\cdots,\beta_{t-1}),$$
where $(\beta_{0},\beta_{1},\cdots,\beta_{t-1})\in\mathbb{Z}_2^{t}$ is a fixed vector.
In addition, we define $L_{0}(\textbf{s}_{n})=R_{0}(\textbf{s}_{n})=\textbf{s}_{n}$ for $t=0$.
\end{defn}

We next present several interesting properties of the nonlinear complexity of $L_{t}(\textbf{s}_{n})$.
\begin{prop}\label{cor1} For a sequence $\textbf{s}_{n}=(s_{0},s_{1},\cdots, s_{n-1})$ over
$\mathbb{Z}_2$ with $nlc(\textbf{s}_{n})=c_{n}\geq\frac{n}{2}$, let $p_{1}$ and $p_{2}$ be
the unique pair of integers such that the condition in (\ref{eq4}) holds and $p_{1}=0$. Then for the $1$-term left extension
sequence $L_{1}(\textbf{s}_{n})=(\alpha_{0},s_{0},s_{1},\cdots, s_{n-1})$,
$nlc(L_{1}(\textbf{s}_{n}))=c_{n}$ if $\alpha_{0}\neq s_{p_{2}-1}$.
\end{prop}
\begin{proof} From Theorem~\ref{thm1}, there exists exactly one pair of identical subsequences of length $c_{n}-1$ with different successors in $\textbf{s}_{n}$, that is,
$$S_{0}^{c_{n}-2} = S_{p_{2}}^{p_{2}+c_{n}-2}\quad \text{and}\quad s_{c_{n}-1}\neq s_{p_{2}+c_{n}-1}.$$
Since $\alpha_{0}\neq s_{p_{2}-1}$, the length of the longest subsequence that occurs twice with different successors in $L_{1}(\textbf{s}_{n})$ is still equal to $c_{n}-1$, then the result follows immediately from Lemma \ref{lem1} (i).
\end{proof}

\begin{thm}\label{thm2} For a sequence $\textbf{s}_{n}=(s_{0},s_{1},\cdots, s_{n-1})$
over $\mathbb{Z}_2$ with $nlc(\textbf{s}_{n})=c_{n}\geq\frac{n+1}{2}$,
let $p_{1}$ and $p_{2}$ be the unique pair of integers such that the condition in (\ref{eq4})
holds and $p_{1}=0$. Let $\textbf{u}_{n+1}=(\alpha,s_{0},s_{1},\cdots, s_{n-1})$ be
the $1$-term left extension sequence of $\textbf{s}_{n}$ with $\alpha\neq s_{p_{2}-1}$ and
$L_{k}(\textbf{u}_{n+1})$ be an arbitrary $k$-term left extension sequence of $\textbf{u}_{n+1}$. Then we have
 $$nlc(L_{k}(\textbf{u}_{n+1}))=c_{n}\quad \text{for}\quad 0\leq k\leq2c_{n}-n-1.$$
\end{thm}

\begin{proof} The result $nlc(L_{0}(\textbf{u}_{n+1}))=c_{n}$ for the case $k=0$ follows
from Proposition \ref{cor1}. It remains to consider the case $1\leq k\leq 2c_n-n-1$. Note that for
the sequence $\textbf{s}_{n}$ we have
\begin{eqnarray*}\label{eqSen}
S_{0}^{c_{n}-2} = S_{p_{2}}^{p_{2}+c_{n}-2}\quad \text{and}\quad s_{c_{n}-1}\neq s_{p_{2}+c_{n}-1}.
\end{eqnarray*}
Set $\textbf{u}_{n+1}=(u_{0},u_{1},\cdots, u_{n})$, where $u_{0}=\alpha$ and $u_{i}=s_{i-1}$ for $1\leq i\leq n$.
Since $nlc(\textbf{u}_{n+1})=nlc(\textbf{s}_{n})=c_{n}\geq\frac{n+1}{2}$, it then follows
from Theorem~\ref{thm1} that there is exactly one pair of identical subsequences of length $c_{n}-1$
with different successors in $\textbf{u}_{n+1}$, that is,
\begin{equation}\label{eqT}
(u_1,u_2,\cdots,u_{c_n-1})=(u_{p_{2}+1},u_{p_{2}+2},\cdots,u_{p_{2}+c_n-1})\quad \text{and}\quad u_{c_n}\neq u_{p_{2}+c_n}.
\end{equation}

Now we define the reciprocal sequence $\widehat{\textbf{u}}_{n+1}=(\widehat{u}_0,\widehat{u}_1,\cdots,\widehat{u}_{n})$
of $\textbf{u}_{n+1}$ by $\widehat{u}_i=u_{n-i}$ for $0\leq i\leq n$.
Then from (\ref{eqT}) we have
\begin{equation*}
(\widehat{u}_{n-c_n+1},\widehat{u}_{n-c_n+2},\cdots,\widehat{u}_{n-1})=
(\widehat{u}_{n-p_{2}-c_n+1},\widehat{u}_{n-p_{2}-c_n+2},\cdots,\widehat{u}_{n-p_{2}-1})
\end{equation*}
and $\widehat{u}_{n-c_n}\neq\widehat{u}_{n-p_{2}-c_n}$. Moreover, $\widehat{u}_{n}=u_0=\alpha$ and $\widehat{u}_{n-p_2}=u_{p_2}=s_{p_2-1}$
yield $\widehat{u}_{n}\neq\widehat{u}_{n-p_{2}}$. We note that all the subsequences of
length $c_n$ in $\textbf{u}_{n+1}$ are distinct by Lemma~\ref{lem1} (iii)
since $nlc(\textbf{u}_{n+1})=c_n\geq \frac{n+1}{2}$. Therefore, the length of the longest subsequence that occurs
twice with different successors in $\widehat{\textbf{u}}_{n+1}$ is $c_n-1$, which implies by
Lemma \ref{lem1} (i) that
\begin{equation}\label{nlc_recip}
 nlc(\widehat{\textbf{u}}_{n+1})=nlc(\textbf{u}_{n+1})=c_{n}.
\end{equation}

Next we consider the $k$-term right extension sequences of $\widehat{\textbf{u}}_{n+1}$. Note that the
length of each sequence $R_{k}(\widehat{\textbf{u}}_{n+1})$ with $0\leq k\leq 2c_n-n-1$
is not greater than $2c_n$. Then by Lemma \ref{lem0} (i) we have
$$nlc(R_{k}(\widehat{\textbf{u}}_{n+1}))=c_{n}\quad \text{for}\quad 0\leq k\leq 2c_n-n-1,$$
and the unique pair of subsequences of length $c_{n}-1$ in $R_{k}(\widehat{\textbf{u}}_{n+1})$ is
the same as that in $\widehat{\textbf{u}}_{n+1}$. Since each sequence $L_{k}(\textbf{u}_{n+1})$
is actually the reciprocal sequence of $R_{k}(\widehat{\textbf{u}}_{n+1})$, the arguments
leading to (\ref{nlc_recip}) show that $nlc(L_{k}(\textbf{u}_{n+1}))=c_n$ for all $1\leq k\leq 2c_n-n-1$.
\end{proof}
\begin{cor}\label{cor2} Let $\textbf{s}_{n}=(s_{0},s_{1},\cdots, s_{n-1})$
be a sequence over $\mathbb{Z}_2$ with $nlc(\textbf{s}_{n})=c_{n}\geq\frac{n}{2}$, and let $p_{1}$ and $p_{2}$ be
the unique pair of integers such that the condition in (\ref{eq4}) holds and $p_{1}\geq1$. Then we have
 $$nlc(L_{k}(\textbf{s}_{n}))=c_{n}\quad \text{for}\quad 1\leq k\leq 2c_n-n.$$
\end{cor}

\begin{proof} We note that $s_{p_{1}-1}\neq s_{p_{2}-1}$, for otherwise there would exist
a pair of identical subsequences of length $c_n$ in $\textbf{s}_n$,
which would contradict Lemma \ref{lem1} (iii). Now we view the sequence $\textbf{s}_{n}$
as $\textbf{u}_{n+1}$, the desired result then follows immediately from Theorem~\ref{thm2}.
\end{proof}

\section{Construction of $\textbf{s}_{n}$ with $nlc(\textbf{s}_{n})\geq\frac{n}{2}$}\label{Sec4}

In this section, we denote the set of all sequences with length $n$ and nonlinear complexity $c_{n}$
over $\mathbb{Z}_{2}$ by $Z_{2}(n,c_{n})$. With the preparations above, we are now ready to present
a construction of all sequences in $Z_{2}(n,c_{n})$ with $c_{n}\geq\frac{n}{2}$, which is based on
a partition of $Z_{2}(n,c_{n})$.

\subsection{A partition of $Z_{2}(n,c_{n})$ with $c_{n}\geq\frac{n}{2}$}\label{Sec4.1}

Recall that there is only a pair of identical subsequences of length $c_{n}-1$ with
different successors in a sequence $\textbf{s}_n$ if $nlc(\textbf{s}_{n})\geq\frac{n}{2}$, and $p_{1}$, $p_{2}$ are
the unique pair of integers with $0\leq p_{1}<p_{2}\leq n-c_{n}$ such that the condition in (\ref{eq4})
holds. Define the distance between the two identical subsequences $S_{p_{1}}^{p_{1}+c_{n}-2}$
and $S_{p_{2}}^{p_{2}+c_{n}-2}$ as
the {\it distance} of $\textbf{s}_{n}$, denoted by $Dis(\textbf{s}_n)$, that is,
$$Dis(\textbf{s}_{n})=p_{2}-p_{1}.$$
It's obvious that $1\leq Dis(\textbf{s}_n)\leq n-c_{n}$. Define a set of sequences
\begin{align*}
Z_{2}(n,c_{n},d)&=\{\textbf{s}_{n}\ | \ \textbf{s}_{n}\in Z_{2}(n,c_{n})\quad \text{and} \quad Dis(\textbf{s}_{n})=d\}.
\end{align*}
Thus, all the sequences in $Z_{2}(n,c_{n})$ are divided into $n-c_{n}$ disjoint
classes according to the distance of each sequence. That is to say, $Z_{2}(n,c_{n})$ can be represented as
\begin{eqnarray}\label{DIRE}
Z_{2}(n,c_{n})=\bigcup\limits^{n-c_{n}}_{d=1}Z_{2}(n,c_{n},d).
\end{eqnarray}
Next we will show how to construct the sets $Z_{2}(n,c_{n},d)$ for $1\leq d\leq n-c_{n}$.

We first consider the case $d=n-c_{n}$, in which the unique pair of
integers $p_{1}$ and $p_{2}$ such that the condition
in (\ref{eq4}) holds for $\textbf{s}_{n}\in Z_{2}(n,c_{n},n-c_{n})$ is specific,
that is, $p_{1}=0$ and $p_{2}=n-c_{n}$.
\subsection{Structure of the sequences in $Z_{2}(n,c_{n},n-c_{n})$}\label{Sec4.2}
To characterize the structure of the sequences in $Z_{2}(n,c_{n},n-c_{n})$,
we give the following definition of periodic (aperiodic) finite-length sequences.

\begin{defn}\label{aper}
Let $\textbf{s}_{n}=(s_{0},s_{1},\cdots, s_{n-1})$ be a finite-length sequence of length $n$
over $\mathbb{Z}_2$. If there exists at least one positive integer $e\,|\,n$ with $e<n$
such that $s_{i+e}=s_i$ for $0\leq i\leq n-e-1$, that is, $\textbf{s}_{n}$ is of the form
$$\textbf{s}_{n}=(s_0,s_1,\cdots,s_{e-1})^{\frac{n}{e}},$$
then $\textbf{s}_n$ is called a periodic finite-length sequence. Otherwise, $\textbf{s}_n$ is called an aperiodic finite-length sequence.
\end{defn}

The following theorem reveals the structure of the sequences in $Z_{2}(n,c_{n},n-c_{n})$.
In what follows, the operation ``$\oplus$'' denotes the addition modulo $2$.

\begin{thm}\label{Main}
Let $n$ and $c_{n}$ be two integers with $n\geq 3$ and $\frac{n}{2}\leq c_n\leq n-1$. Then a sequence $\textbf{s}_{n}=(s_{0},s_{1},\cdots,s_{n-1})\in Z_{2}(n,c_{n},n-c_{n})$ if and only if $\textbf{s}_{n}$ has the form
\begin{eqnarray}\label{FORM}
\textbf{s}_{n}=\big((s_0,s_1,\cdots,s_{n-c_{n}-1})^m(s_0,s_1,\cdots,s_{r-1},\overline{s}_{r})\big),
\end{eqnarray}
where $(s_0,s_1,\cdots,s_{n-c_{n}-1})$ is an arbitrary aperiodic
finite-length sequence over $\mathbb{Z}_2$, $m\geq1$ and $0\leq r<n-c_{n}$ are integers such that $n-1=(n-c_{n})m+r$, and $\overline{s}_{r}=s_{r}\oplus1$.
\end{thm}

\begin{proof} The case $n-c_n=1$ is trivial by Lemma \ref{lem1} (ii). For the case $n-c_n>1$, we first show the necessity.
Suppose that $\textbf{s}_{n}=(s_0,s_1,\cdots,s_{n-1})\in Z_{2}(n,c_{n},n-c_{n})$.
By Remark \ref{mark1}, $Dis(\textbf{s}_{n})=p_2-p_1=n-c_{n}$ suggests $p_{1}=0$ and $p_{2}=n-c_{n}$. Thus
$$S_{0}^{c_{n}-2} = S_{n-c_{n}}^{n-2}\quad \text{and}\quad s_{c_{n}-1}\neq s_{n-1},$$
which is equivalent to
$$s_{i}=s_{i+(n-c_{n})}\quad\text{for}\quad0\leq i\leq c_{n}-2\quad\text{and}\quad s_{c_{n}-1}\oplus1=s_{n-1}.$$
And since $n-1=(n-c_{n})m+r$, it follows that $\textbf{s}_n$ has the form
\begin{eqnarray*}
\textbf{s}_{n}=\big((s_0,s_1,\cdots,s_{n-c_{n}-1})^m(s_0,s_1,\cdots,s_{r-1},\overline{s}_{r})\big),
\end{eqnarray*}
where $\overline{s}_{r}=s_{r}\oplus1$.
Next, we claim that $(s_{0},s_1,\cdots,s_{n-c_{n}-1})$ must be an aperiodic finite-length sequence.
Suppose that $(s_{0},s_1,\cdots,s_{n-c_{n}-1})$ is a periodic finite-length sequence
with length $n-c_{n}\geq 2$. Then there exists a positive divisor $e$ of $n-c_{n}$
with $e<n-c_{n}$ such that $s_{i+e}=s_i$ holds for each $i$ with $0\leq i\leq n-c_{n}-e-1$ according to Definition \ref{aper}.
This implies that $(s_{0},s_1,\cdots,s_{n-c_{n}-1})$ has the form $(s_{0},s_1,\cdots,s_{e-1})^f$, where the integer $f=(n-c_{n})/e$. Together with (\ref{FORM}), we obtain that $\textbf{s}_n$ has the form
$$\textbf{s}_n=\big((s_0,s_1,\cdots,s_{e-1})^{mf+q}(s_0,s_1,\cdots,s_{r'-1},\overline{s}_{r'})\big),$$
where $q$ and $r'$ with $r'<e$ are the unique pair of nonnegative integers satisfying $r=eq+r'$,
and $\overline{s}_{r'}=s_{r'}\oplus1$. Thus, one can verify that in the sequence $\textbf{s}_n$,
$$S_{0}^{n-e-2}=S_{e}^{n-2}\quad \text{and}\quad s_{n-e-1} \neq s_{n-1},$$
which suggests $nlc(\textbf{s}_n)\geq n-e$ by Lemma \ref{lem1} (i), and hence
$nlc(\textbf{s}_n)>c_{n}$, a contradiction.

Conversely, if the sequence $\textbf{s}_n$ is of the form (\ref{FORM}),
then one can verify that there exists a pair of identical subsequences of
length $(n-c_{n})(m-1)+r=c_{n}-1$ with different successors in $\textbf{s}_n$, that is,
\begin{eqnarray*}
S_{0}^{c_{n}-2} = S_{n-c_{n}}^{n-2}\quad \text{and}\quad s_{c_{n}-1}\neq s_{n-1}.
\end{eqnarray*}
Thus, we get $nlc(\textbf{s}_{n})\geq c_{n}$ by Lemma \ref{lem1} (i) and note that the distance of this pair of identical subsequences is $n-c_n$. To show $\textbf{s}_{n}\in Z_{2}(n,c_{n},n-c_{n})$ it suffices to prove $nlc(\textbf{s}_{n})=c_{n}$. We do this by showing that all the subsequences of length $c_{n}$ in $\textbf{s}_{n}$ are distinct. Suppose that there exists a subsequence of length $c_{n}$ that occurs twice with successors in $\textbf{s}_{n}$, say
$$S_{a}^{a+c_{n}-1} = S_{b}^{b+c_{n}-1},$$
where $0\leq a<b\leq n-c_{n}-1$. Put $\delta=b-a$, then we have $1\leq\delta\leq n-c_{n}-1$ and
\begin{eqnarray}\label{eq32}
s_{\rho}=s_{\rho+\delta}\quad \mbox{for}\quad a\leq\rho\leq a+c_{n}-1.
\end{eqnarray}
Since $c_{n}\geq\frac{n}{2}$, we have $c_{n}\geq n-c_{n}$. From (\ref{eq32}) we know that $s_{\rho}=s_{\rho+\delta}$
must hold for consecutive $n-c_{n}$ terms in $\textbf{s}_{n}$, and therefore
\begin{eqnarray}\label{Th32}
 s_{i}=s_{i+\delta}\quad \mbox{for}\quad 0\leq i\leq n-c_{n}-1.
\end{eqnarray}
Let $\tau=\gcd\,(\delta,n-c_{n})$, and let $u\in\mathbb{Z}_{n-c_{n}}$ and $v\in \mathbb{Z}_\delta$
be the integers such that $u\delta-v(n-c_{n})=\tau$.
Now we consider a periodic sequence $\textbf{t}^{n-c_n}=(t_0,t_1,\cdots,t_{n-c_n-1},\cdots)$ which is completely specified by the first subsequence
of length $n-c_n$ in $\textbf{s}_n$, that is, $t_i=s_{i\,(\text{mod}\,n-c_n)}$ for all $i\geq 0$.
From the definition of a periodic sequence and (\ref{Th32}) we have
$$t_{i}=t_{i+n-c_{n}} \quad \mbox{and} \quad t_{i}=t_{i+\delta}$$
hold for $0\leq i\leq n-c_{n}-1$. This implies that
$$t_{i}=t_{i+v(n-c_{n})}\quad \mbox{and} \quad t_{i}=t_{i+u\delta}$$
for all $i\geq 0$. Then with $\tau=u\delta-v(n-c_{n})$ and $t_i=s_i$ for $0\leq i\leq n-c_n-1$, we get $s_{i}=s_{i+\tau}$ for $0\leq i\leq n-c_{n}-1$. This contradicts the
fact that $(s_0,s_1,\cdots,s_{n-c_{n}-1})$ is an aperiodic finite-length sequence, and the proof is complete.
\end{proof}

\begin{remark}\label{mark2}
Theorem \ref{Main} shows that each sequence $\textbf{s}_n\in Z_{2}(n,c_{n},n-c_{n})$ is completely
determined by its first subsequence of length $n-c_{n}$, and can be further obtained by employing an
aperiodic finite-length sequence of length $n-c_{n}$. Let $\textbf{s}_{n}=(s_0,s_1,\cdots,s_{n-1})$ be
the binary sequence defined by
\begin{equation}\label{CON}
s_i=
\begin{cases}
t_{i\,(\text{mod}\,n-c_{n})},&\text{if $0\leq i\leq n-2$,}\\
t_{i\,(\text{mod}\,n-c_{n})}\oplus1,&\text{if $i=n-1$,}
\end{cases}
\end{equation}
where $\textbf{t}_{n-c_{n}}=(t_0,t_1,\cdots,t_{n-c_{n}-1})$ is an aperiodic finite-length sequence
in $\mathbb{Z}_2^{n-c_{n}}$. Then we can generate all the sequences in  $Z_{2}(n,c_{n},n-c_{n})$ by
letting $\textbf{t}_{n-c_{n}}$ run through all aperiodic finite-length sequences of length $n-c_{n}$.
\end{remark}

\subsection{The extension method for $Z_{2}(n,c_{n},d)$ with $1\leq d< n-c_{n}$}\label{Sec4.3}

By what we have already shown in Subsection \ref{Sec4.2}, all sequences
in $Z_{2}(n,c_{n},d)$ with $d=n-c_{n}$ can be constructed by applying the method presented in (\ref{CON}).
In this subsection, we propose a method to generate all sequences in $Z_{2}(n,c_{n},d)$ with $1\leq d<n-c_{n}$ by extending each sequence in $Z_{2}(c_{n}+d,c_{n},d)$ with $n-(c_{n}+d)$ terms, where $2\leq n-c_{n}\leq \frac{n}{2}$.

Let $\textbf{s}_{c_{n}+d}$ be an arbitrary sequence in $Z_{2}(c_{n}+d,c_{n},d)$, and
let $t_{1}$ and $t_{2}$ be a pair of nonnegative
integers satisfying $t_{1}+t_{2}=n-(c_{n}+d)$. An extension method to generate all sequences
in $Z_{2}(n,c_{n},d)$ proceeds as follows:
\begin{equation}\label{STEP}
\left\{\begin{aligned}
&\textrm{STEP 1}.\,\, \mathfrak{L}_{t_{1}}(\textbf{s}_{c_{n}+d}): \textbf{s}_{c_{n}+d}=(s_{0},s_{1},\cdots,s_{c_{n}+d-1})\rightarrow
L_{t_{1}}(\textbf{s}_{c_{n}+d})=(\alpha_{t_{1}-1},\cdots,\alpha_{0},\textbf{s}_{c_{n}+d});&\\
&\textrm{STEP 2}.\,\, \mathfrak{R}_{t_{2}}(L_{t_{1}}(\textbf{s}_{c_{n}+d})):
L_{t_{1}}(\textbf{s}_{c_{n}+d})\rightarrow
R_{t_{2}}(L_{t_{1}}(\textbf{s}_{c_{n}+d}))=(L_{t_{1}}(\textbf{s}_{c_{n}+d}),\beta_{0},\cdots,\beta_{t_{2}-1}),&\\
\end{aligned}\right.
\end{equation}
where $L_{t_{1}}(\textbf{s}_{c_{n}+d})$ denotes any $t_1$-term left extension sequence of $\textbf{s}_{c_{n}+d}$
such that $\alpha_{0}=s_{d-1}\oplus1$ and $(\alpha_{1},\alpha_2,\cdots, \alpha_{t_1-1})\in\mathbb{Z}_2^{t_1-1}$ if $t_{1}\geq 1$,
and $R_{t_{2}}(L_{t_{1}}(\textbf{s}_{c_{n}+d}))$ denotes any $t_2$-term right extension sequence of $L_{t_{1}}(\textbf{s}_{c_{n}+d})$ such that $(\beta_{0},\beta_1,\cdots, \beta_{t_2-1})\in\mathbb{Z}_2^{t_2}$ if $t_{2}\geq 1$.

For each pair of fixed nonnegative integers $t_{1}$ and $t_{2}$ with $t_{1}+t_{2}=n-(c_{n}+d)$,
let $\mathcal{B}_{t_{1},t_{2}}(Z_{2}(c_{n}+d,c_{n},d))$ denote the set of all
sequences obtained by extending each sequence in $Z_{2}(c_{n}+d,c_{n},d)$ according to the extension method in (\ref{STEP}).
Then we have the following proposition.

\begin{prop}\label{thmcon} Let $n$, $c_n$ and $d$ be integers such that $2\leq n-c_{n}\leq \frac{n}{2}$ and $1\leq d\leq n-c_{n}-1$.
Then we have
\begin{center}
  $Z_{2}(n,c_{n},d)=\bigcup\limits_{t_{1}+t_{2}=n-(c_{n}+d)}\mathcal{B}_{t_{1},t_{2}}(Z_{2}(c_{n}+d,c_{n},d))$.
\end{center}
\end{prop}
\begin{proof}

For each pair of nonnegative integers $t_1$ and $t_2$ with $t_{1}+t_{2}=n-(c_{n}+d)$, a sequence
$\textbf{e}_{n}\in \mathcal{B}_{t_{1},t_{2}}(Z_{2}(c_{n}+d,c_{n},d))$ has the following general form
$$\textbf{e}_n=R_{t_{2}}(L_{t_{1}}(\textbf{s}_{c_{n}+d}))=(\alpha_{t_{1}-1},\cdots,\alpha_{0},s_{0},s_{1},\cdots,s_{c_{n}+d-1},\beta_{0},\beta_1,\cdots, \beta_{t_2-1}),$$
where $(s_{0},s_{1},\cdots,s_{c_{n}+d-1})$ is a certain sequence in $Z_{2}(c_{n}+d,c_{n},d)$, and %$(\alpha_{1},\alpha_2,\cdots, \alpha_{t_1-1})\in\mathbb{Z}_2^{t_1-1}$ if $t_{1}\geq 1$,
%and $R_{t_{2}}(L_{t_{1}}(\textbf{s}_{c_{n}+d}))$
$\alpha_i$ with $0\leq i\leq t_1-1$ and $\beta_j$ with $0\leq j\leq t_2-1$ satisfy the conditions in the extension method in (\ref{STEP}).

Since $\textbf{s}_{c_{n}+d}\in Z_{2}(c_{n}+d,c_{n},d)$ and $c_n\geq \frac{n}{2}\geq \frac{c_n+d+1}{2}$, the unique pair of integers such that the condition in (\ref{eq4}) holds for $\textbf{s}_{c_{n}+d}$ is obviously $p_1=0$ and $p_2=d$.
Then by Theorem \ref{thm2} the nonlinear complexity of $L_{t_{1}}(\textbf{s}_{c_{n}+d})$ is equal to $c_n$ since $\alpha_{0}\neq s_{d-1}$. Furthermore, for the sequence $\textbf{e}_n=R_{t_{2}}(L_{t_{1}}(\textbf{s}_{c_{n}+d}))$,
%As for the sequence $R_{t_{2}}(L_{t_{1}}(\textbf{s}_{c_{n}+d}))$,
$c_n\geq \frac{n}{2}=\frac{c_n+d+t_1+t_2}{2}$ and Lemma \ref{lem0} (i) ensure that its nonlinear complexity is still equal to $c_n$.
%Then it follows from Theorem~\ref{thm2} and Lemma \ref{lem0} (i) that $nlc(\textbf{e}_{n})=c_{n}$.
On the other hand, because $c_n\geq \frac{n}{2}$, there is only a pair of identical subsequences of length $c_n-1$ with different successors in $\textbf{e}_n$ and it is exactly the only pair of such subsequences in $\textbf{s}_{c_n+d}$,
which implies $Dis(\textbf{e}_n)=Dis(\textbf{s}_{c_{n}+d})=d$.
Thus we have $\textbf{e}_{n}\in Z_{2}(n,c_{n},d)$, and hence
$$\bigcup\limits_{t_{1}+t_{2}=n-(c_{n}+d)}\mathcal{B}_{t_{1},t_{2}}(Z_{2}(c_{n}+d,c_{n},d))\subset Z_{2}(n,c_{n},d).$$

Conversely, for any sequence $\textbf{s}_{n}\in Z_{2}(n,c_{n},d)$, from Theorem~\ref{thm1} and (\ref{eq4}), there exists a unique pair of integers $p_{1},p_{2}$ with $0\leq p_{1}<p_{2}\leq n-c_{n}$ and $p_{2}-p_{1}=d$ such that
$$S_{p_{1}}^{p_{1}+c_{n}-2} = S_{p_{2}}^{p_{2}+c_{n}-2}\quad \text{and}\quad s_{p_{1}+c_{n}-1}\neq s_{p_{2}+c_{n}-1}.$$
Since $c_n\geq\frac{n}{2}$, all the subsequences of length $c_n$ in $\textbf{s}_n$ are
distinct according to Lemma \ref{lem1} (iii). Therefore, $s_{p_{1}-1}\neq s_{p_{2}-1}$ and the subsequence
$$S_{p_{1}}^{p_{2}+c_{n}-1}=(s_{p_{1}},s_{p_{1}+1},\cdots, s_{p_{2}+c_{n}-1})$$
of $\textbf{s}_{n}$ itself is a sequence of length $c_{n}+d$ with nonlinear complexity $c_{n}$ and distance $d$, that is to say, $S_{p_{1}}^{p_{2}+c_{n}-1}\in Z_{2}(c_{n}+d,c_{n},d)$. It implies that $\textbf{s}_{n}\in\mathcal{B}_{t_1,\,t_{2}}(Z_{2}(c_{n}+d,c_{n},d))$ with $t_1=p_1$ and $t_2=n-c_{n}-p_{2}$, and hence
 $$Z_{2}(n,c_{n},d)\subset \bigcup\limits_{t_{1}+t_{2}=n-(c_{n}+d)}\mathcal{B}_{t_{1},t_{2}}(Z_{2}(c_{n}+d,c_{n},d)).$$
This completes the proof of Proposition \ref{thmcon}.
\end{proof}

Proposition \ref{thmcon} shows that if $\textbf{s}_{c_{n}+d}$ runs through all the sequences in $Z_{2}(c_{n}+d,c_{n},d)$ and $t_1$, $t_2$ run through all the nonnegative integers satisfying $t_{1}+t_{2}=n-(c_{n}+d)$, then all the sequences in $Z_{2}(n,c_{n},d)$ are obtained by the extension method presented in (\ref{STEP}).

Based on the construction proposed in Section \ref{Sec4}, we are able to propose an algorithm that generates all binary sequences with given length $n$ and nonlinear complexity $c_n\geq\frac{n}{2}$.
\begin{algorithm}[H]
	%\floatname{algorithm}{Algorithm}%change algorithm name
	\renewcommand{\algorithmicrequire}{\textbf{Input:}}%change algorithm input
	\renewcommand{\algorithmicensure}{\textbf{Output:}}%change algorithm output
	%\footnotesize
	\caption{Construction of binary sequences $\textbf{s}_{n}$ with $nlc(\textbf{s}_{n})\geq\frac{n}{2}$.}
	\label{alg:Construction}
	\begin{algorithmic}[1]
		\Require Two positive integers $n$ and $c_{n}$ with $\frac{n}{2}\leq c_{n}\leq n-1$.
		
		\Ensure $Z_{2}(n,c_{n})$ and its cardinality.
		\State $Z_{2}(n,c_{n})\gets\emptyset$; $count\gets0$  \Comment{initialize $Z_{2}(n,c_{n})$ and its cardinality}
		\For {$d=1$ to $n-c_{n}$}
		\For {$(t_{0},t_1,\cdots,t_{d-1})\in\mathbb{Z}_2^{d}$}
        \While {$(t_{0},t_1,\cdots,t_{d-1})$ is an aperiodic finite-length sequence}
        \For {$i=0$ to $c_{n}+d-2$} \Comment{rule in (\ref{CON})}
        \State $s_{i}\gets t_{i\,(\text{mod}\,d)}$
        \EndFor
        \State $s_{c_{n}+d-1}\gets t_{c_{n}+d-1\,(\text{mod}\,d)\,}\oplus1$
        \State $\textbf{s}_{c_{n}+d}\gets(s_{0},s_1,\cdots,s_{c_{n}+d-1})$
        \For {$m=0$ to $n-(c_{n}+d)$} \Comment{extension method in (\ref{STEP})}
        \If {$m=0$}
        \State $\textbf{s}_{n}\gets R_{n-(c_{n}+d)}(L_{0}(\textbf{s}_{c_{n}+d}))$
        \Else
        \State $\textbf{s}_{n}\gets R_{n-(c_{n}+d)-m}(L_{m-1}(s_{d-1}\oplus1,\textbf{s}_{c_{n}+d}))$
        \EndIf
        \State append the sequence $\textbf{s}_{n}$ to $Z_{2}(n,c_{n})$
        \State $count\gets count+1$
        \EndFor
        \EndWhile
        \EndFor
        \EndFor

		\State \Return the set: $Z_{2}(n,c_{n})$; the cardinality of $Z_{2}(n,c_{n})$: $count$
		
	\end{algorithmic}
\end{algorithm}
\begin{example}\label{example2}

To illustrate our construction, we give an example for the case $n=8$ and $c_{n}=4$.
By first using the method presented in (\ref{CON}), we can obtain all the sequences in $Z_2(5,4,1)$, $Z_2(6,4,2)$,
$Z_2(7,4,3)$ and $Z_2(8,4,4)$, and then by applying the extension method given in (\ref{STEP}), each sequence in $Z_2(5,4,1)$, $Z_2(6,4,2)$ and $Z_2(7,4,3)$ is extended to several
sequences in $Z_2(8,4,1)$, $Z_2(8,4,2)$ and $Z_2(8,4,3)$, respectively.
Thus, we get a total of $86$ sequences of length $8$ and nonlinear complexity $4$,  and all these sequences are listed in Table \ref{Tab01}. Our result is consistent
with that of exhaustive search given in \cite[Table~3.1]{Jansen1}.

\begin{table}[!h]
\caption{Binary sequences of length $8$ with nonlinear complexity $4$.}
\scriptsize
\begin{center}
\begin{tabular}{|l|l|l|l| p{1cm}|}
\hline
\hline
$d$ & $Z_{2}(d+4,4,d)$ & $Z_{2}(8,4,d)$   &$\mathcal{N}$\\ \hline
\hline
\multirow{10}{*}{1} &  \multirow{5}{*}{00001}
  & (00001)000, (00001)001, (00001)010, (00001)011  & \multirow{10}{*}{40}  \\

 && (00001)100, (00001)101, (00001)110, (00001)111  &  \\
 && 1(00001)00, 1(00001)01, 1(00001)10, 1(00001)11  &  \\
 && 01(00001)0, 01(00001)1, 11(00001)0, 11(00001)1  &  \\
 && 001(00001), 011(00001), 101(00001), 111(00001)  &  \\
 \cline{2-3}
 &\multirow{5}{*}{11110}
  & (11110)111, (11110)110, (11110)101, (11110)100  &  \\
 && (11110)011, (11110)010, (11110)001, (11110)000  &  \\
 && 0(11110)11, 0(11110)10, 0(11110)01, 0(11110)00  &  \\
 && 10(11110)1, 10(11110)0, 00(11110)1, 00(11110)0  &  \\
 && 110(11110), 100(11110), 010(11110), 000(11110)  &  \\
 \hline
\multirow{4}{*}{2} &  \multirow{2}{*}{010100}
  & (010100)00, (010100)01, (010100)10, (010100)11  & \multirow{4}{*}{16}\\
 && 0(010100)0, 0(010100)1, 00(010100), 10(010100)  &  \\

\cline{2-3}
& \multirow{2}{*}{101011}
  & (101011)11, (101011)10, (101011)01, (101011)00  &  \\
 && 1(101011)1, 1(101011)0, 11(101011), 01(101011)  &  \\

\hline

\multirow{6}{*}{3} &  0010011
& (0010011)0, (0010011)1, 0(0010011)
 &\multirow{6}{*}{18}\\

\cline{2-3}
& 1101100 & (1101100)1, (1101100)0, 1(1101100) & \\
\cline{2-3}
& 0100101 & (0100101)0, (0100101)1, 1(0100101) & \\

\cline{2-3}
& 1011010 & (1011010)1, (1011010)0, 0(1011010) &\\
\cline{2-3}
& 0110111 & (0110111)0, (0110111)1, 0(0110111) &\\
\cline{2-3}
& 1001000 & (1001000)1, (1001000)0, 1(1001000) &\\
 \hline

\multirow{2}{*}{4} &  \multirow{2}{*}{$Z_{2}(8,4,4)$}
 & 00010000, 00100011, 00110010, 01000101, 01100111, 01110110     & \multirow{2}{*}{12}\\
&& 11101111, 11011100, 11001101, 10111010, 10011000, 10001001     &  \\

\hline
\hline
\end{tabular}
\end{center}
\label{Tab01}
\end{table}

\end{example}
\section{Enumeration and distribution of nonlinear complexity}\label{Sec5}

In this section, we shall count the number of all sequences with length $n$ and
nonlinear complexity $c_n\geq \frac{n}{2}$. To this end, we first study the
cardinality of $Z_{2}(n,c_{n},n-c_{n})$.

From Theorem \ref{Main} and Remark \ref{mark2}, $|Z_{2}(n,c_{n},n-c_{n})|$ is equal to the number of all aperiodic finite-length sequences of length $n-c_{n}$, so it only depends on the value of $n-c_n$.
The following proposition gives the exact formula for $|Z_{2}(n,c_n,n-c_n)|$, which is obtained by subtracting the number of periodic finite-length sequences from the total number $2^{n-c_{n}}$ of sequences of length $n-c_{n}$.

\begin{prop}\label{propNum} Let $n$ and $c_{n}$ be two positive integers with $\frac{n}{2}\leq c_{n}\leq n-1$.

{\rm (i) }
If $n-c_{n}=1$, then $|Z_{2}(n,c_{n},n-c_{n})|=2$.

{\rm (ii) }
If $n-c_{n}>1$, let $n-c_{n}=\prod\limits^{t}_{j=1}p_{j}^{k_{j}}$ be the standard factorization of $n-c_{n}$,
where $p_1,p_2,\cdots,$ $p_t$ are distinct prime numbers and $k_1,k_2,\cdots,k_t$ are positive integers. Then we have
$$|Z_{2}(n,c_{n},n-c_{n})|=2^{n-c_{n}}-\sum\limits^{t}_{\tau=1}(-1)^{\tau-1}\sum\limits_{1\leq j_{1}<j_{2}<\cdots<j_{\tau}\leq t}2^{\frac{n-c_{n}}{p_{j_{1}}p_{j_{2}}\cdots p_{j_{\tau}}}}.$$
\end{prop}
\begin{proof}
The case $n-c_{n}=1$ is trivial by Lemma \ref{lem1} (ii).
For the case $n-c_{n}>1$, we determine $|Z_{2}(n,c_{n},n-c_{n})|$ by means of the principle of inclusion and exclusion in Combinatorial Mathematics \cite{LIU}.

Let $U$ be the set of all finite-length sequences of length $n-c_{n}$ over $\mathbb{Z}_2$.
It is obvious that $|U|=2^{n-c_{n}}$. Now we consider the property of a periodic finite-length sequence.
From Definition \ref{aper}, we know that a periodic finite-length sequence $\textbf{s}_{n-c_{n}}$
must have the form
\begin{eqnarray}\label{INEX}
\textbf{s}_{n-c_{n}}=(s_0,s_1,\cdots,s_{e-1})^{\frac{n-c_{n}}{e}},
\end{eqnarray}
where $e$ is a positive divisor of $n-c_{n}$ with $e<n-c_{n}$.

For $1\leq j\leq t$, let $X_{j}$ denote the property that a sequence $\textbf{s}_{n-c_{n}}$ has the form in (\ref{INEX})
with $\frac{n-c_{n}}{e}=p_{j}$, and let $P(X_{j})$ denote the set of all the sequences in $U$ that
possess the property $X_j$. Then it follows that $|P(X_{j})|=2^e=2^{\frac{n-c_{n}}{p_j}}$.

Let $P(X_{j_{1}},X_{j_{2}},\cdots,X_{j_{\tau}})$ be the set of all sequences that
possess each of the properties $X_{j_{1}},X_{j_{2}},\cdots,X_{j_{\tau}}$, defined as
$$P(X_{j_{1}},X_{j_{2}},\cdots,X_{j_{\tau}})=\bigcap\limits_{j\in\{j_{1},j_{2},\cdots,j_{\tau}\}}P(X_{j}).$$
Since $p_{j_1}, p_{j_2},\cdots, p_{j_\tau}$ are distinct prime numbers, it follows that
$P(X_{j_{1}},X_{j_{2}},\cdots,X_{j_{\tau}})$ consists of all the sequences having the form
in (\ref{INEX}) with $\frac{n-c_{n}}{e}=p_{j_{1}}p_{j_{2}}\cdots p_{j_{\tau}}$, and thus we have
\begin{eqnarray}\label{INEX2}
|P(X_{j_{1}},X_{j_{2}},\cdots,X_{j_{\tau}})|=2^{\frac{n-c_{n}}{p_{j_{1}}p_{j_{2}}\cdots p_{j_{\tau}}}}.
\end{eqnarray}
Applying the principle of inclusion and exclusion, we obtain
\begin{eqnarray*}
|Z_{2}(n,c_{n},n-c_n)|&=&|\bigcap\limits_{1\leq j\leq t}\overline{P}(X_{j})|\\
            &=&|U|-|\bigcup\limits_{1\leq j\leq t}P(X_{j})|\\
            &=&2^{n-c_n}-\sum\limits_{1\leq j_{1}\leq t}|P(X_{j_{1}})|+\sum\limits_{1\leq j_{1}< j_{2}\leq t}|P(X_{j_{1}},X_{j_{2}})|+\cdots+\\
            &&(-1)^{\tau}\sum\limits_{1\leq j_{1}< \cdots< j_{\tau}\leq t}|P(X_{j_{1}},\cdots,X_{j_{\tau}})|
            +\cdots+(-1)^{t}|P(X_{1},X_{2},\cdots,X_{t})|.
\end{eqnarray*}
Together with (\ref{INEX2}), the desired result follows.
\end{proof}
Now we can further determine the cardinality of $Z_{2}(n,c_{n})$.
The following theorem gives a formula for $|Z_{2}(n,c_{n})|$.

\begin{thm}\label{Numb} Given two positive integers $n$ and $c_{n}$ with $\frac{n}{2}\leq c_{n}\leq n-1$, the total number of sequences with length $n$ and
nonlinear complexity $c_{n}$ is given by
\begin{center}
  $|Z_{2}(n,c_{n})|=\sum\limits^{n-c_{n}}_{d=1}(n-c_{n}-d+2)2^{n-c_{n}-d-1}|Z_{2}(c_{n}+d,c_{n},d)|$.
\end{center}
%where $|Z_{2}(c_{n}+d,c_{n},d)|$ is calculated by the formula given in Proposition \ref{propNum}.
\end{thm}

\begin{proof}
We recall from Subsection \ref{Sec4.3} that for each pair of fixed nonnegative integers $t_{1}$ and $t_{2}$
with $t_{1}+t_{2}=n-(c_{n}+d)$, $\mathcal{B}_{t_{1},t_{2}}(Z_{2}(c_{n}+d,c_{n},d))$ denotes the set of all
the sequences obtained by extending each sequence in $Z_{2}(c_{n}+d,c_{n},d)$ according to the extension method in (\ref{STEP}).
Note that the number of sequences generated by extending a fixed sequence $\textbf{s}_{c_n+d} \in Z_{2}(c_{n}+d,c_{n},d)$
is $2^{t_1+t_2}=2^{n-(c_n+d)}$ if $t_1=0$ and $2^{t_1+t_2-1}=2^{n-(c_n+d)-1}$ otherwise, because $\alpha_{i}$ and $\beta_{j}$, where  $1\leq i\leq t_1-1$ and $0\leq j\leq t_2-1$, can take any element in $\mathbb{Z}_2$. Thus we get
\begin{equation}\label{number}
|\mathcal{B}_{t_{1},t_{2}}(Z_{2}(c_{n}+d,c_{n},d))|=
\begin{cases}
2^{n-(c_{n}+d)}|Z_{2}(c_{n}+d,c_{n},d)|,&\text{if $t_{1}=0$,}\\
2^{n-(c_{n}+d)-1}|Z_{2}(c_{n}+d,c_{n},d)|,&\text{if $t_{1}\neq0$.}
\end{cases}
\end{equation}

In addition, for $(t'_1,t'_2)\neq(t_1,t_2)$, we must have
\begin{eqnarray}\label{disjoint}
\mathcal{B}_{t'_{1},t'_{2}}(Z_{2}(c_{n}+d,c_{n},d))\cap\mathcal{B}_{t_{1},t_{2}}(Z_{2}(c_{n}+d,c_{n},d))=\emptyset.
\end{eqnarray}
For otherwise, a sequence $\textbf{v}_n\in \mathcal{B}_{t'_{1},t'_{2}}(Z_{2}(c_n+d,c_n,d))\cap \mathcal{B}_{t_{1},t_{2}}(Z_{2}(c_n+d,c_n,d))$
would imply that there are two pairs of identical subsequences of length $c_n-1$ with different
successors in $\textbf{v}_n$, a contradiction.
Combining (\ref{DIRE}), (\ref{number}), (\ref{disjoint}) and Proposition~\ref{thmcon}, we arrive at
\begin{eqnarray*}
|Z_{2}(n,c_{n})|&=&\sum\limits_{d=1}^{n-c_{n}}|Z_{2}(n,c_{n},d)|\\
            &=&|Z_{2}(n,c_n,n-c_{n})|+\sum\limits_{d=1}^{n-c_{n}-1}
            \sum\limits_{t_{1}+t_{2}=n-(c_{n}+d)}|\mathcal{B}_{t_{1},t_{2}}(Z_{2}(c_{n}+d,c_{n},d))|\\
                        &=&|Z_{2}(n,c_n,n-c_{n})|+\sum\limits_{d=1}^{n-c_{n}-1}
            {[2^{n-c_n-d}+(n-c_n-d)2^{n-c_n-d-1}]|Z_{2}(c_{n}+d,c_{n},d)|}\\
            &=&\sum\limits_{d=1}^{n-c_{n}}(n-c_{n}-d+2)2^{n-c_{n}-d-1}|Z_{2}(c_{n}+d,c_{n},d)|.
\end{eqnarray*}
\end{proof}

From the formula in Theorem \ref{Numb}, we find that for fixed $n-c_n$, the
cardinality of $Z_{2}(n,c_{n})$ is determined by each $|Z_{2}(c_{n}+d,c_{n},d)|$ with $1\leq d\leq n-c_n$.
Since $d\leq n-c_n\leq c_n$, it follows that $c_n\geq \frac{c_n+d}{2}$, and so $|Z_{2}(c_{n}+d,c_{n},d)|$
can be calculated by the formula in Proposition \ref{propNum}.
As we have noted, $|Z_{2}(c_{n}+d,c_{n},d)|$ is uniquely determined by $d$, so that $|Z_{2}(n,c_{n})|$ is uniquely determined by the value of $n-c_{n}$.
Thus we get the following result on the distribution of the nonlinear complexity for finite-length binary sequences under the condition $c_{n}\geq\frac{n}{2}$.

\begin{cor}\label{Distr} Let $n$ and $c_{n}$ be two positive integers with $\frac{n}{2}\leq c_{n}\leq n-1$.
Then for any integer $\delta\geq n-2c_n$, we have
\begin{eqnarray}\label{distrib}
  |Z_{2}(n+\delta,c_{n}+\delta)|=|Z_{2}(n,c_{n})|.
\end{eqnarray}
\end{cor}
\begin{proof}
With $\delta\geq n-2c_n$, we get $2(c_{n}+\delta)-(n+\delta)=2c_{n}-n+\delta\geq0$,
which implies $c_{n}+\delta\geq\frac{n+\delta}{2}$. Therefore, by Theorem~\ref{Numb}, we have
$$|Z_{2}(n+\delta,c_{n}+\delta)|=\sum\limits^{n-c_{n}}_{d=1}(n-c_{n}-d+2)2^{n-c_{n}-d-1}|Z_{2}(c_{n}+\delta+d,c_{n}+\delta,d)|.$$
Since $c_{n}\geq\frac{n}{2}$, it follows that
$$|Z_{2}(n,c_{n})|=\sum\limits^{n-c_{n}}_{d=1}(n-c_{n}-d+2)2^{n-c_{n}-d-1}|Z_{2}(c_{n}+d,c_{n},d)|.$$
From Proposition \ref{propNum}, we have $|Z_{2}(c_{n}+\delta+d,c_{n}+\delta,d)|=|Z_{2}(c_{n}+d,c_{n},d)|$
for a fixed $d$, and the desired result follows immediately.
\end{proof}

\begin{remark}\label{Evenl}
The result of Corollary \ref{Distr} is equivalent to that of Proposition 3.30 in \cite{Jansen1} and of
Theorem 11 in \cite{RK}. But we prove this property in a different way.
By substituting $\delta$ by $n-2c_n$ in (\ref{distrib}), we get $$|Z_{2}(n,c_{n})|=|Z_{2}(2(n-c_{n}),n-c_{n})|.$$
Therefore, to determine the exact value of $|Z_{2}(n,c_{n})|$ for all integers
$n\geq 3$ and $c_n\geq \frac{n}{2}$, we only need to calculate $|Z_{2}(n,\frac{n}{2})|$ for all even $n$.
\end{remark}

%By applying the formulas in  Theorem \ref{Numb} and Proposition \ref{propNum}, we calculate the value of %$|Z_{2}(n,\frac{n}{2})|$ for even $n$ with $2\leq n\leq 40$.

\begin{example}\label{example3}
By applying the formulas in Theorem \ref{Numb} and Proposition \ref{propNum}, we calculate the value of $|Z_{2}(n,\frac{n}{2})|$ for even $n$ with $2\leq n\leq 48$, which are given in Table \ref{distab}.
From the value of $|Z_{2}(n,\frac{n}{2})|$ listed in Table \ref{distab}, one can obtain the value of $|Z_{2}(n+\delta,\frac{n}{2}+\delta)|$ for all integers $\delta$.
It is verified that for length $n=2,4,\cdots,24$, our result is consistent with the result obtained by the exhaustive search in \cite[Table 3.1]{Jansen1}.

\begin{table}[!h]
			 \caption {The cardinality of $Z_{2}(n,\frac{n}{2})$ with even $n$.}
             \footnotesize
             %\scriptsize
			 \label{distab}
			 \begin{center}	
			 \begin{tabular}{|c|c||c|c||c|c||c|c| p{1cm}|}
			 \hline
			 \hline
			 $n$ & $|Z_{2}(n,\frac{n}{2})|$  & $n$  & $|Z_{2}(n,\frac{n}{2})|$  & $n$ & $|Z_{2}(n,\frac{n}{2})|$& $n$ & $|Z_{2}(n,\frac{n}{2})|$\\
					\hline
					\hline
					 $2$  & $2$     & $14$  & $1792$      & $26$    & $363874$  &  $38$  & $48642118$\\
					 \hline
					 $4$  & $8$     & $16$   & $4562$     & $28$    & $839312$  &  $40$  & $107569382$\\
                     \hline
					 $6$  & $28$    & $18$   & $11344$    & $30$    & $1918228$ &  $42$  & $236758534$\\
					 \hline
					 $8$  & $86$    & $20$   & $27614$    & $32$    & $4348214$ &  $44$  & $518851844$\\
					 \hline
					 $10$  & $250$  & $22$   & $66136$    & $34$    & $9785734$ &  $46$  & $1132569592$\\
					 \hline
					 $12$  & $680$  & $24$   & $156062$   & $36$    & $21880586$&  $48$  & $2463255266$\\
					
                     \hline
					 \hline
			 \end{tabular}
			\end{center}
\end{table}
\end{example}

\begin{remark}
 Erdmann and Murphy in \cite{MANN} proposed an approximate probability distribution for nonlinear complexity.
The accuracy of the approximation for $n\leq 24$ was confirmed by the numerical results in \cite[Table 3.1]{Jansen1}.
However, the accuracy for large length is unknown.
Proposition \ref{propNum} and Theorem \ref{Numb} in this paper present a theoretical result on the number of finite-length binary sequences with length $n$ and nonlinear complexity $c_n$.
 The exact value of
 $|Z_{2}(n,c_{n})|$ with $\frac{n}{2}\leq c_n\leq n-1$ for a very large integer $n$ can be calculated, as long as the standard factorization of each integer less than or equal to $n-c_n$ is known. The exact probability that a random binary sequence of length $n$ achieves the nonlinear complexity $c_{n}$ is thereby obtained by
  $Pr[nlc(\textbf{s}_{n})=c_n]=\frac{|Z_{2}(n,c_{n})|}{2^{n}}$ for each $c_n$ with $\frac{n}{2}\leq c_n\leq n-1$.
\end{remark}

\section{Conclusion}\label{Sec6} In this paper, we proceeded with theoretical investigation of
binary sequences with nonlinear complexity not less than half of the length. Based on the
structural properties, we provided, for the first time, a direct construction of binary sequences
with given length $n$ and nonlinear complexity $c_{n}\geq n/2$. Moreover, a
formula was established to calculate the exact number of these sequences.
An interesting future work would be the extension of the design techniques to the non-binary case.
%Besides, another main open problem is the construction of periodic
%sequences with large nonlinear complexity (not less than half of the period).
This problem deserves further research.

\section*{Acknowledgment}
The work of Liang, Zeng and Sun was supported by the National Natural Science Foundation of China (NSFC) under Grant $62072161$, and by Application Foundation Frontier Project of Wuhan Science and Technology Bureau under Grant $2020010601012189$.
The work of Xiao was supported by the National Natural Science Foundation of China under Grant 12061027.

\end{document}